\documentclass[lettersize,journal]{IEEEtran}
\usepackage{amsmath,amsfonts}
\usepackage{algorithmic}
\usepackage{algorithm}
\usepackage{array}
\usepackage[caption=false,font=normalsize,labelfont=sf,textfont=sf]{subfig}
\usepackage{textcomp}
\usepackage{stfloats}
\usepackage{url}
\usepackage{verbatim}
\usepackage{graphicx}
\usepackage{cite}
\usepackage{hyperref}

\usepackage{amsthm}
\usepackage{enumitem}
\usepackage{algorithmic}
\usepackage{xcolor}
\usepackage{booktabs}
\usepackage{bm}
\usepackage{float}
\usepackage{multirow}
\usepackage{comment}
\usepackage{booktabs}
\usepackage{algorithm}
\usepackage{amssymb}

\newtheorem{theorem}{Theorem}
\newtheorem{lemma}{Lemma}

\newtheorem{corollary}{Corollary}

\newtheorem{definition}{Definition}

\begin{document}

\title{Community Detection in the \\ Multi-View Stochastic Block Model}

\author{Yexin Zhang, Zhongtian Ma, Qiaosheng Zhang, Zhen Wang, Xuelong Li
        % <-this % stops a space
\thanks{
Yexin Zhang and Zhongtian Ma are with the School of Cybersecurity and also with the School of Artificial Intelligence, OPtics, and ElectroNics (iOPEN), Northwestern Polytechnical University, Xi’an Shaanxi, 710072, P.R.China (e-mail: yexin.zhang@mail.nwpu.edu.cn, mazhongtian@mail.nwpu.edu.cn).

Qiaosheng Zhang is with the Shanghai Artificial Intelligence Laboratory, Shanghai 200232, P.R.China (e-mail: zhangqiaosheng@pjlab.org.cn).

Zhen Wang and Xuelong Li are with the Key Laboratory of Intelligent Interaction and Applications, Ministry of Industry and Information Technology, and the School of iOPEN, Northwestern Polytechnical University, Xi’an Shaanxi, 710072, P.R.China (email: w-zhen@nwpu.edu.cn, li@nwpu.edu.cn). Corresponding authors: Zhen Wang and Qiaosheng Zhang.

}
% <-this % stops a space
}


\markboth{ }%
{Shell \MakeLowercase{\textit{et al.}}: A Sample Article Using IEEEtran.cls for IEEE Journals}


\maketitle

\begin{abstract}
This paper considers the problem of community detection on multiple potentially correlated graphs from an information-theoretical perspective. We first put forth a random graph model, called the multi-view stochastic block model (MVSBM), designed to generate correlated graphs on the same set of nodes (with cardinality $n$). The $n$ nodes are partitioned into two disjoint communities of equal size. The presence or absence of edges in the graphs for each pair of nodes depends on whether the two nodes belong to the same community or not. The objective for the learner is to recover the hidden communities with observed graphs. Our technical contributions are two-fold: (i) We establish an information-theoretic upper bound (Theorem~1) showing that exact recovery of community is achievable when the model parameters of MVSBM exceed a certain threshold. (ii) Conversely, we derive an information-theoretic lower bound (Theorem~2) showing that when the model parameters of MVSBM fall below the aforementioned threshold, then for any estimator, the expected number of misclassified nodes will always be greater than one. Our results for the MVSBM recover several prior results for community detection in the standard SBM as well as in multiple independent SBMs as special cases.  
\end{abstract}


\section{Introduction}
\IEEEPARstart{C}{ommunity} detection stands as a foundational and focal point of research in the realm of network science\cite{newman2003structure}. The intricate web of relationships, whether they pertain to social networks, biological systems, or information propagation, often exhibits underlying structures composed of cohesive groups of nodes or individuals\cite{gao2023multilayer,Lelarge2015Labelled}. Identifying and characterizing these communities within networks not only enriches our understanding of complex systems but also holds practical applications in fields as diverse as recommendation systems\cite{ahn2018binary,zhang2021community,zhang2022mc2g}, biology\cite{sporns2014contributions}, and  sociology\cite{newman2004detecting}.

One prominent approach to community detection is the application of the Stochastic Block Model (SBM)\cite{HOLLAND1983109}, a probabilistic generative model that reflects the underlying structure of networks. Specifically, SBMs assume the nodes are separated into several communities, and the connections between nodes are determined by the community assignments, which indicate the communities to which the two nodes belong\cite{abbe2015exact,abbe2015community}.  
The parameters of SBM include the number of communities, the size of each community, and the probabilities of connections between nodes. Nodes within the same community typically have a higher probability of being connected compared to nodes in different communities. The investigation of SBMs not only brings valuable theoretical insights but also plays a crucial role in the design of algorithms and their practical applications.

In recent years, efforts have been made to enhance the understanding of community detection and SBMs\cite{abbe2015community, chin2015stochastic, gao2017achieving,saad2018community,abbe2020entrywise,esmaeili2021semidefinite,esmaeili2022community, zhang2022exact,chien2019minimax,sima2021exact}. 
However, most prior works on SBMs only consider performing community detection on a \emph{single graph}.
In practice, data relationships do not exist in isolation but can be analyzed and interpreted from \emph{multiple views}, each offering a unique perspective and complementary information. For instance, multi-view graphs are used in multi-lingual information retrieval to address data sparsity and language discrepancies. In the context of community detection, each text is viewed as a graph with the information points in texts serving as nodes of the graph, and the connections between information points form the edges of the multi-view graphs. This allows for cross-lingual information fusion and improved retrieval performance. Multi-view data is also widely used in computer vision, including group behavior analysis\cite{li2018multiview} and pedestrian detection\cite{tu2020voxelpose}. With an increase in the number of camera setups, multi-view images and multi-view feature similarity maps are essential for comprehensive data analysis, especially in crowded scenes. In the context of community detection, images from different cameras with varying degrees are considered as multi-view graphs, where a single pedestrian is a node, for instance, in group behavior analysis, while a group of pedestrians is regarded as a community. Additionally, the connection between nodes can be modeled by social attributes such as movement direction or gatherings.

Motivated by these applications, in this paper, we investigate the use of multi-view graphs and introduce a new random graph model for community detection, referred to as the Multi-view SBM (MVSBM). The MVSBM randomly generates $D \in \mathbb{N}^+$ potentially correlated graphs simultaneously. These graphs share the same set of nodes (of cardinality $n \in \mathbb{N}^+$), and these nodes are assumed to belong to two balanced communities (of equal size)\footnote{As an initial effort to study community detection on multiple correlated graphs, we mainly focus on the most basic setting where there are two balanced communities. This is a common and representative setting in the SBM literature and has been extensively studied in prior works~\cite{abbe2015exact,mossel2015consistency,saad2018community,abbe2020entrywise,esmaeili2021semidefinite,esmaeili2022community,sima2021exact,jog2015information}. In future work, it is worthwhile extending our investigation to more general settings with multiple communities of non-equal sizes.}. For each pair of nodes, the probability of them being connected in one graph depends on whether they belong to the same community or not (like the standard SBM) and also depends on their connection in other graphs (since the $D$ graphs are correlated). We direct the readers to Section~\ref{sec:MVSBM} for a  comprehensive description of the MVSBM. In essence, the MVSBM encompasses diverse perspectives and features from multi-views to enhance community detection analysis, address data heterogeneity, capture complementary information, improve robustness, and enable a holistic understanding of the data.
Technically, we investigate the fundamental constraints of community detection in the MVSBM from an information-theoretic perspective, by establishing upper and lower bounds on the task of recovering communities.

% In future work, it is worthwhile extending our investigation to more general settings with multiple communities of non-equal sizes.

\subsection{Main contributions}

The paper's main contributions can be summarized as follows.
\begin{itemize}
    \item Firstly, we introduce a new random graph model, called the \emph{multi-view stochastic block model} (MVSBM), that can generate multiple graphs on the same set of nodes. Each graph represents the relationships among these nodes from a distinct view. Different graphs are generated in a correlated manner.  %each representing a graph acquired from a distinct perspective of the same object. 
    % While the problem of community detection from multiple graphs (or multi-view clustering) has been studied in fields such as machine learning and computer vision, there remains a lack of theoretical understanding in this domain. 
    The MVSBM offers a theoretical framework for investigating the community detection problem from multiple graphs/views, thereby providing valuable insights for practical applications.

    \item We establish an information-theoretic upper bound for community detection in the MVSBM (Theorem 1). Specifically, we focus on the task of \emph{exact recovery} of communities (defined in Definition~\ref{def:exact}), and show that when the model parameters of the MVSBM satisfy the condition in Theorem~\ref{T1}, then exact recovery is achievable by using the maximum likelihood estimator.
    
    \item We also derive an information-theoretic lower bound that serves as an impossibility result for \emph{any} estimator (Theorem 2). Specifically, we show that when the model parameters of the MVSBM fail to meet the condition in Theorem~\ref{T1}, then for any estimator, the expected number of misclassified nodes will always be greater than one. 
    
    \item  
Comparing Theorem~\ref{T1} with Theorem~\ref{T2}, we observe the presence of a \emph{sharp threshold} in the MVSBM (as indicated by Eqn. (10)), such that exact recovery of communities is possible when above the threshold, while all estimators must misclassify at least one node when below the threshold.
% This threshold plays a crucial role in understanding the MVSBM, in the sense that exact recovery of communities is possible when above the threshold, while all estimators must misclassify  (on average) at least one node when below the threshold. 
Furthermore, when specializing the MVSBM to several special cases (such as when the multiple graphs are identical), we observe that our results coincide with the previous results derived for these special cases.
%It is important to highlight that the exact recovery upper and lower bounds we have derived are closely aligned, thereby establishing a sharp threshold for exact recovery.

 %and derive an upper bound for \emph{exact recovery},i.e., correct recovery for all points. More precisely, exact recovery can be achieved when the parameters of MVSBM satisfy the conditions in Equation~6. A detailed description of this can be found in Theorem 1.

\end{itemize}
\subsection{Related works}

Since our work focuses on the dovetail between multi-view graphs and community detection, we will highlight the most pertinent related works on these two topics in this subsection.

There is extensive literature on identifying hidden community structures in networks, especially on the celebrated SBMs. Abbé's survey\cite{abbe2017community} provides an overview of community detection in the SBM, including the threshold for exact community recovery in the SBM\cite{abbe2015exact,mossel2015consistency,mossel2015reconstruction}. We further develop and enhance these analyses, particularly focusing on integrating complementary information and enhancing the resilience of community detection across multi-view graphs.
%We highlight in particular the works of Abbé, Bandeira, and Hall [2] and Mossel, Neeman, and Sly [42], 
%There are also some works that go beyond SBMs, and incorporate additional graph information to aid in recovering communities, like multi-layer networks SBMs~\cite{ma2023community,chen2022global} and correlated SBMs~\cite{racz2021correlated, gaudio2022exact}.

Some works consider the diversity of connections or relationships between nodes and model them as \emph{multilayer networks}, where each type of connection is treated as a network layer~\cite{kivela2014multilayer}. Therefore, it has given rise to several complementary models for multi-layer SBMs~\cite{ma2023community,chen2022global,stanley2016clustering,lei2023computational,chatterjee2022clustering}. Based on the underlying community structure, a collection of graphs are generated on the same set of nodes with identical latent community labels. Multi-layer networks tend to capture the dynamic changes or hierarchical structures of networks, such as the evolving community structure in social networks over time\cite{yang2011detecting} or the multi-level modular organization in biological networks~\cite{gosak2018network}.  Several variants have been explored, but typically, given the community labels, the layers are \emph{conditionally independent}~\cite{ma2023community,de2017community,stanley2016clustering,lei2023computational,chatterjee2022clustering}. On the other hand, the approach in this paper is to propose a probabilistic model for detecting the community structure of from multiple correlated graphs, where the graphs refer to multiple graphical representations obtained from different perspectives or modalities, which are used to handle network data from diverse data sources or with different features. This is a significant difference from the settings we are considering.

% For instance, in social networks, there can be user relationship graphs and interest graphs, or in protein-protein interaction networks, there can be graphs generated from different experimental techniques. Multi-layer networks primarily focus on the hierarchical structure or temporal evolution of networks, while multi-view graphs primarily emphasize the multiple perspectives or modalities of networks. Compared with multi-layer SBM, 

There are also other works that consider learning latent community structure from multiple correlated graphs~\cite{racz2022correlated,racz2021correlated,gaudio2022exact,onaran2016optimal}. Ref.~\cite{racz2021correlated} explores exact community recovery by determining the information-theoretic limit for exact graph matching in correlated SBMs. Ref.~\cite{gaudio2022exact} examines the information-theoretic threshold for exact community recovery from two correlated block models, which specifically deals with uncertainties as well as dependencies resulting from partial and inexact matching between the correlated graphs. This is a notable contrast to the MVSBM in which we do not consider graph matching as an intermediate step for community detection. The MVSBM tackles the challenge of integrating and merging data from various views in a more comprehensive manner.

In particular, we highlight the works of the labeled SBM~\cite{yun2016optimal} and weighted SBM~\cite{jog2015information}, which are closely related to this paper. The labeled SBM denotes the observation of the similarity between two nodes in the form of a label, then inferring the hidden partition from observing the random labels on each pair of nodes. In the MVSBM, the connection vector (across $D$ graphs) between a pair of nodes can be considered as a label in the labeled SBM. In fact, the proof of our information-theoretic lower bound is inspired by the proof in~\cite{yun2016optimal}, particularly their change-of-measure technique. Ref.~\cite{jog2015information} reveals a connection between the R\'{e}nyi divergence and the success probability of the maximum likelihood estimator for exact community recovery in the weighted SBM. While the problem in~\cite{jog2015information} is similar to ours, the proof techniques are different. %However, instead of generating random labels independently based on the community membership of nodes in a single graph, we utilize joint probability distributions across multiple views to model the process of graph generation.

% \subsection{Organization}
% This paper is structured as follows. We introduce the MVSBM as well as three special cases of the MVSBM in Section II. Then we provide our main theoretical results along with three corollaries in section III. Section IV contains the proof of Theorem 1, and Section V presents the proof of Theorem 2. We conclude the paper in Section VI.

\vspace{10pt}
\section{Preliminaries and problem setup}
\subsection{Notations}
For positive integers $a, b \in \mathbb{N}$ such that $a < b$, let $[a: b]$
denote the set of integers $\{a, a+1, \ldots, b\}$ ranging from a to $b$. Let $[b]$ denote the set of integers $\{1, 2,\ldots, b \}$. Random variables are denoted by uppercase letters (e.g., $\bm{X}$), while their realiations are denoted by lowercase letters (e.g., $\bm{x}$). We use calligraphic letters to denote sets (e.g., $\mathcal{X}$). For any set $\mathcal{X}$, let $\Delta(\mathcal{X})$ be the \emph{probability simplex} (i.e., all possible probability distributions) over $\mathcal{X}$.

We adopt \emph{asymptotic notations}, including $O(.)$, $o(.)$, $\Omega(.)$, $\omega(.)$, and $\Theta(.)$, to describe the limiting behaviour of functions/sequences~\cite[Chapter~3.1]{cormen2022introduction}. For instance, we say a pair of functions $f(n)$ and $g(n)$ satisfies $f(n) = O(g(n))$ if there exist $m > 0$ and $N_0 \in \mathbb{N}^+$ such that for all $n > N_0$, $|f(n)| \le mg(n)$. For an event $\mathcal{E}$, we use $\mathbf{1}\{\mathcal{E} \}$ to denote the \emph{indicator function} that outputs $1$ if $\mathcal{E}$ is true and outputs $0$ otherwise. Additional notations and definitions used are provided in Table~\ref{table:1}.

%\ma{ For a positive integer $n$, we define $[n] := \{1, 2,\ldots, n \}$. }

\subsection{Mulit-view Stochastic Block Model (MVSBM)} \label{sec:MVSBM}
% \ye{In this paper, we consider a symmetric setting with $n$ nodes and two equal-sized communities $C_0$ and $C_1$. For each node \( i \), let \( \delta(i) \in \{-1, +1\} \) denote its community assignment, i.e., if the node $i$ is in the community $C_0$, then the node $i$ is assigned $\delta(i) = -1$, and if the node $i$ is in the community $C_1$, then the node $i$ is assigned $\delta(i) = +1$. The underlying partition of the communities can be denoted as $\bm{\delta} = ( \delta(1), \delta(2), \ldots, \delta(n) )$, which is drawn uniformly from the set of  $ \Delta := \{ \delta \in \{-1, +1\}^n : \sum_{i=1}^{n} \delta(i) = 0 \} $  of all possible community assignments. And $\delta$ and $-\delta$ correspond to the same balanced partition.}

In this paper, we consider a symmetric setting with $n$ nodes partitioned into two equal-sized communities: $C_+$ and $C_-$. Each node $i \in [n]$ is labeled with $X(i) \in \{-1, +1\}$, signifying its membership in either community $C_+$ or community $C_-$. The ground truth label of the $n$ nodes is denoted by a length-$n$ vector $\bm{X} = (X(1), X(2), \ldots, X(n))$, which is sampled uniformly 
 at random from the set 
\begin{align}
 \mathcal{\bm{X}} := \left\{ \bm{x} \in \{-1, +1\}^n : \sum_{i=1}^{n} x(i) = 0 \right\} 
 \end{align}
 that contains all possible ground truth labels. Note that $\bm{x}$ and $-\bm{x}$ correspond to the same balanced partition of the $n$ nodes.

A multi-view stochastic block model (MVSBM) comprises $D$ random graphs, represented by $\bm{G} := \{G^1, G^2, \ldots, G^D \}$, where $D$ is a positive integer. Each graph offers a distinct ``view'' on the same set of nodes $[n]$, and all the graphs share the same ground truth label. For each graph $G^d$ (where $d \in [D]$), we denote its corresponding adjacency matrix as $A^d \in \{0, 1\}^{n \times n}$. In this matrix, $A^d_{ij} = 1$ denotes the presence of an edge connecting nodes $i$ and $j$, while $A^d_{ij} = 0$ denotes the absence of an edge between nodes $i$ and $j$. We further define the collection of the $D$ adjacency matrices as an \emph{adjacency tensor} $\bm{A} := [A^1, A^2, \ldots, A^D]$. As each $A^d$ is an $n \times n$ matrix, we note that the adjacency tensor $\bm{A}$ exhibits dimensions of $D \times n \times n$. By noting that $\bm{A}$ and $\bm{G}$ are in a one-to-one mapping relationship, one can use the adjacency tensor $\bm{A}$ to represent $\bm{G}$, as illustrated in the Fig.\ref{fig:picture}.

Intuitively, graphs from different views exhibit commonalities, but each view provides unique detailed information. To capture this property, we assume the $D$ random graphs in $\bm{G}$ are generated in a correlated manner and may not be independent. Specifically, the MVSBM employs joint probability distributions across the $D$ views to model the graph generation process, as elaborated in the following.
  
For each pair of nodes $(i, j)$, let $\bm{A}_{ij} := (A^1_{ij}, A^2_{ij},\ldots, A^D_{ij} ) $ be the \emph{connection vector} between nodes $i$ and $j$ across the $D$ views. Specifically, if nodes $i$ and $j$ belong to the same community, it is assumed that the connection vector $A_{ij}$  follows a distribution $p \in \Delta(\{0,1 \}^D)$, where $\Delta(\{0,1 \}^D)$ denotes the probability simplex over the set $\{0,1\}^D$. On the other hand, if nodes $i$ and $j$ belong to different communities, it is assumed that $A_{ij}$  follows a different distribution $q \in \Delta(\{0,1 \}^D)$. Here, $p$ and $q$ are specific representations of the joint probability distribution of the $D$ random graphs.

% For a given pair of $p$ and $q$, we denote the corresponding mulit-view stochastic block model as MVSBM$(n, p, q)$. Using this model, we can generate the adjacency tensor $\bm{A}$, which corresponds to the collection of the $D$ graphs within $\bm{G}$.

\begin{definition}[MVSBM]
    Given positive integers $n$ and $D$, and a pair of distributions  
    $p, q \in \Delta(\{ 0, 1 \}^D)$, 
    the MVSBM$(n, p, q, D)$ is a random graph model that samples the ground truth label $\bm{X}=(X(1), X(2), \ldots, X(n))$ as well as the $D$ random graphs with adjacency tensor $\bm{A}$ according to the following rules: 
    \begin{itemize}
        \item $\bm{X}$ is chosen uniformly at random from $\mathcal{\bm{X}}$.
        \item The connection vector $\bm{A}_{ij}$ is generated according to
        \begin{align}
        \bm{A}_{ij} \sim \begin{cases}
            p, \quad \text{if } X(i)  = X(j), \\
            q, \quad \text{if }  X(i)  \ne X(j).
        \end{cases} \label{eq:A}
        \end{align}

    \end{itemize}
\end{definition}

\begin{table}[t]\centering \caption{Notations and definitions}
% \begin{tabular}{c|l}
\begin{tabular}{p{1cm}|p{7cm}}
\hline
\multicolumn{1}{l|}{Notation} & Definition                                                               \\ \hline
$n$                             &  number of nodes                                                      \\ \hline
$p$                             & probability distribution of connection (in the same community)        \\ \hline
$q$                             & probability distribution of connection (in different communities) \\ \hline
$D $                            &  number of  graphs                                          \\ \hline
$\bm{A}$                           &  adjacency tensor of $n\times n \times D$ dimensions                                 \\ \hline
$A^d$                           &  adjacency matrix of the $d$-th graph                                   \\ \hline
$\bm{A}_{ij}$                           & connection vector of nodes $i$ and  $j$ across $D$ graphs               \\ \hline
$A^d_{ij} $                        &  connection of nodes $i$ and $j$ in the $d$-th graph    \\ \hline
$\mathcal{X}$                            & the community labels space                                               \\ \hline
$\bm{X}$                             & ground truth community label (random variable)                                 \\ \hline
$\bm{x}$                             & ground truth community label (realization)                                 \\ \hline
$X(i)$                         & the $i$-th element of $\mathcal{X}$ (random variable)                                      \\ \hline
$x(i)$                         & the $i$-th element of $\bm{X}$ (realization)                                      \\ \hline
$\bm{x}_0$                           & a specific  community label $\{+1,\cdots, +1,-1,\cdots, -1\}$                                       \\ \hline
\(E\)                            & error event (the estimator fails to output the ground truth)      \\ \hline
\(\varepsilon_{\phi}(n) \)                         & number of nodes misclassified by the estimator $\phi$                               \\ \hline
\(I(p,q) \)                         & R\'{e}nyi divergence between $p$ and $q$ (of order $1/2$)                            
\label{table:1}    \\ \hline
\end{tabular}
\end{table}

\subsection{Objectives}

Based on the adjacency tensor $\bm{A}$ generated according to the MVSBM, our objective is to use an estimator $\phi=\phi(\bm{A})$  to exactly recover the underlying partition $\bm{x}$ or $-\bm{x}$ (since $\bm{x}$ and $-\bm{x}$ correspond to the same partition). The estimator outputs estimated labels belonging to the set $\mathcal{\bm{X}}$.  In Definition~\ref{def:exact} below, we present the criterion for achieving the exact recovery of communities in the MVSBM.

\begin{definition}[Exact Recovery in the \textup{MVSBM}$(n, p, q,D)$] \label{def:exact}
    We say that \emph{exact recovery} of communities in the MVSBM$(n,p,q,D)$ is achievable if there exists an estimator $\phi(\bm{A})$ such that the following holds:
    $$
    \lim_{n \to \infty}  \underset{(\bm{X} , \bm{A}) \sim \text{MVSBM}(n, p, q, D) }{\mathbb{P}}\big(\phi(\bm{A}) \!=\! \bm{X} \text{ or } \phi(\bm{A}) \!=\! -\bm{X} \big) = 1.
    $$
    It essentially requires that the probability of successfully recovering the underlying partition tends to one as the number of nodes~$n$ tends to infinity.
 
    % On the contrary, we say exact recovery is impossible if the above statement fails for \emph{any} estimator. 
\end{definition}

\begin{figure}[t] \centering 
\includegraphics[width=8.24cm,height=2.5cm]{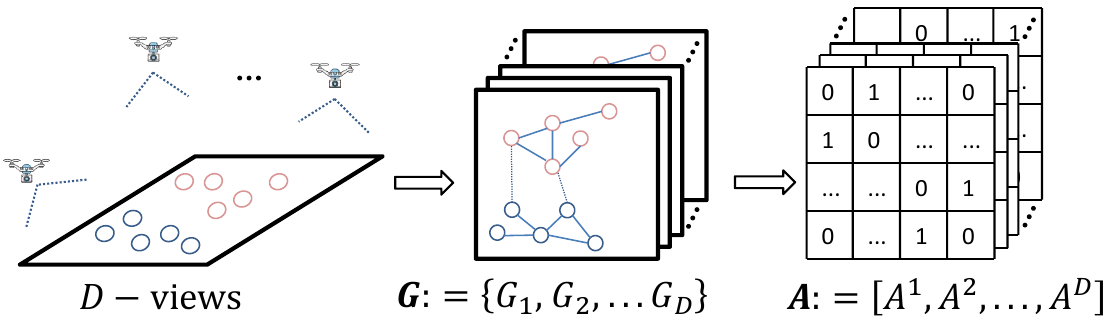} \caption{Illustration of the multi-view stochastic block model (MVSBM) that comprises $D$ view graphs.
} \label{fig:picture} 
\end{figure}

\subsection{Special Cases of MVSBM} \label{sec:special}

The MVSBM can be seen as a generalization of the standard SBM, expanding its scope from analyzing a single graph to multiple graphs that may exhibit correlations. With the ability to cover several random graph models as special cases (as described below), the MVSBM becomes a versatile model for investigating a wide range of community detection problems.

\subsubsection{Special Case 1}
The standard SBM with two balanced communities~\cite{abbe2015exact} is a special case of the MVSBM when specializing \(D=1\) (i.e., when there is only a single view). In this case, the distributions $p$ and $q$ are both Bernoulli distributions, satisfying $p(1) + p(0) = 1$ and $q(1) + q(0) = 1$ respectively. Each pair of nodes is connected with probability $p(1)$ if they belong to the same community, and is connected with probability $q(1)$ otherwise. The probabilities that a pair of nodes is not connected are thus either $p(0)$ or $q(0)$, depending on whether they belong to the same community or not.   %the probabilities of having an edge or not between them are \(p(1)\) and \(p(0)\) respectively. Conversely, if they belong to different communities, the corresponding probabilities are \(q(1)\) and \(q(0)\).

\subsubsection{Special Case 2}
The MVSBM has the capability to encompass the extreme case in which the $D$ graphs are \emph{identical}, i.e., for any pair of nodes, if an edge is present (resp. absent) in one graph, it will also be present (resp. absent) in all other graphs. The MVSBM  degenerates to  this case by choosing the distributions $p$ and $q$ as 

%Considering a case, a node pair having edges in all \(D\) views is said to be ``connected''. We can express the probabilities of being connected within the same community \(p\) and between different communities \(q\) in terms of \(\bm{d}\) as follows:
\begin{equation} \label{eq:sp1}
p(\bm{d})=\left\{
\begin{array}{c}
\begin{aligned}
&\alpha \quad\quad\quad\quad  \bm{d}=(1,1, \ldots, 1), \\
&1-\alpha \ \quad\quad \bm{d}=(0,0, \ldots, 0), \\ 
&0 \quad\quad\quad\quad  \text{otherwise},
\end{aligned}
\end{array}
\right.
\end{equation}
\begin{equation} \label{eq:sp2}
q(\bm{d})=\left\{
\begin{array}{c}
\begin{aligned}
&\beta \quad\quad\quad\quad  \bm{d}=(1,1, \ldots, 1), \\
&1-\beta \ \quad\quad \bm{d}=(0,0, \ldots, 0), \\ 
&0 \quad\quad\quad\quad  \text{otherwise},
\end{aligned}
\end{array}
\right.
\end{equation}
for some $\alpha, \beta \in [0,1]$.
Here, \(\bm{d}=(1,1, \ldots, 1)\) represents the case where an edge is present in all the $D$ graphs, while \(\bm{d}=(0,0, \ldots, 0)\) represents the absence of edges in all the $D$ graphs. Intuitively, observing multiple identical graphs is equivalent to observing a single graph in this special case. It is also not hard to see that the MVSBM with distributions $p$ and $q$ given in~\eqref{eq:sp1} and~\eqref{eq:sp2} is equivalent to the standard SBM described in Special Case~1.

\subsubsection{Special Case 3}

In addition to the extreme case of multiple identical graphs, the MVSBM can also encompass the opposite extreme case where the $D$ graphs are \emph{independent} (each is generated via an independent SBM). This can be achieved by choosing the distributions $p$ and $q$ to be \emph{product distributions}. Specifically, for every $i \in [D]$, let $p_i$ and $q_i$ be Bernoulli distributions that represent the probabilities of observing an edge in the $i$-th graph. The distributions $p$ and $q$ can then be expressed as 
\begin{align}
    p(\bm{d}) = \prod_{i=1}^{D} p_i(d_i) \quad \text{and} \quad q(\bm{d}) = \prod_{i=1}^{D} q_i(d_i), \label{eq:sp3}
\end{align}
where $d_i$ is the $i$-th element in the connection vector $\bm{d}$.

\vspace{10pt}
\section{Main results and Corollaries}
In this section, we present our main results regarding the conditions under which conditions the exact recovery of communities is achievable in the MVSBM.
Before introducing these results, we first introduce a key quantity $I(p,q)$ that measures the separation between the two distributions $p$ and $q$ in the MVSBM: 
% \(D(p \| q)\) as given by follows, where \( t \in [0,1] \), 
% \begin{align}
%     D(p \| q) : = \max_{t \in [0,1]} &\left[ \log \frac{1}{\sum_{\bm{d} \in \{0,1\}^D} p(\bm{d})^t q(\bm{d})^{1-t}} \right. \nonumber \\
%     &\phantom{=} \left.  + \log \frac{1}{\sum_{\bm{d} \in \{0,1\}^D} p(\bm{d})^{1-t} q(\bm{d})^t} \right]
%     \label{eq:divergence}
% \end{align}
%The R\^{e}nyi divergence of order \(1/2\) is given by
\begin{equation}
I(p, q) := -2 \log \left( \sum_{\bm{d} \in \{0,1\}^D} [p(\bm{d}) q(\bm{d}]^{1/2} \right).
\end{equation}
Note that $I(p,q)$ is essentially the \emph{R\'{e}nyi divergence} (of order $1/2$) between $p$ and~$q$.
%where \(\bm{d} \in \{0,1\}^D\) is a connection vector across the \(D\) views, and  \(p(\bm{d})\) (resp. \(q(\bm{d})\)) is the probability of \(\bm{d}\) occurring when a pair of nodes is within the same community (resp. in different communities) and $p(\bm{d}) > q(\bm{d})$. 
Intuitively, the two communities are easier to be separated if the divergence $I(p,q)$ between the distributions $p$ and $q$ is larger, and are harder to be separated otherwise.

\subsection{Assumptions}

Let $0^D := (0, 0, ..., 0)$ be an abbreviation of the length-$D$ all-zero vector, and let 
\begin{align}
\bar{p} := \max_{\bm{d} \in \{0,1\}^D \setminus \{ 0^D \} } \max\{p(\bm{d}), q(\bm{d}) \} 
\end{align}
be the maximum probability of observing at least one edge across the $D$ graphs (i.e., observing a connection vector $\bm{d}$ that does not equal $0^D$) between any pair of nodes. Similar to most prior works on the SBM, we consider the parameter regime where the probability of observing an edge is a vanishing function of $n$ (i.e., it is assumed that $\bar{p} = o(1)$), while we also assume $\bar{p} = \omega(1/n)$ to exclude the scenario wherein the graphs are extremely sparse (in which case exact recovery of communities is known to be impossible).  Moreover, we introduce two additional assumptions (A1) and (A2) on the distributions $p$ and $q$ in the following.

%Conditionally on the connection relationship of nodes to communities, the $\bm{d}=(0, 0, ..., 0)$ is the most frequent connection assignment over $D$ views. We assume that $\bar{p} = o(1)$ and $\bar{p}n = \omega(1)$. And \(\bar{p} = \max_{\bm{d} \in \{0,1\}^D} p(\bm{d})\) represents the maximum probability of a pairwise connection under any assignment \(\bm{d}\) instead of $\bm{d}=(0, 0, ..., 0)$.

\begin{itemize}
    \item[--] (A1) There exists a positive constant $\rho > 1$ such that
\[
\frac{1}{\rho}\le \frac{p(\bm{d})}{q(\bm{d})} \leq \rho, \qquad \forall \bm{d} \in \{0,1\}^D.
\]
This assumption requires that the differences between the probabilities of observing any connection vector $\bm{d} \in \{0,1\}^D$ under distributions $p$ and $q$ are not substantial.
\item[--] (A2) There exists a positive constant \(\varepsilon > 0\) such that 
\[
\frac{ \sum_{\bm{d}\in \{0,1\}^D \setminus \{0^D\} }[p(\bm{d})-q(\bm{d})]^2}{\bar{p}^2} \geq \varepsilon.
\]
This assumption imposes a certain separation between the two distributions $p$ and $q$. 
\end{itemize}

\subsection{Main Results}\label{Results}
We first present a sufficient condition for exact recovery in the MVSBM (in Theorem~\ref{T1} below), where the condition depends critically on the R\'{e}nyi divergence $I(p,q)$. 

\begin{theorem}\label{T1}
    When the model parameters of the MVSBM$(n,p,q,D)$ satisfy
    \begin{equation}
    \lim_{n \to \infty}\frac{nI(p,q)}{\log n} > 2,
    \end{equation}
then exact recovery is achievable.
\end{theorem}

In particular, we show that Theorem~\ref{T1} can be achieved by using the maximum likelihood (ML) estimator $\phi_{\text{ML}}$, which ensures $\mathbb{P}(\phi_{\text{ML}}(\bm{A}) \in 
\{\bm{X}, -\bm{X} \} ) = 1 - o(1)$ and thus satisfies the exact recovery criterion defined in Definition~\ref{def:exact}. The detailed proof of Theorem~\ref{T1} is provided in Section~\ref{sec:proof_thm1}.

Next, we provide a \emph{negative result} (Theorem~\ref{T2}) showing that if the model parameters of the MVSBM fail to meet the condition in Theorem~\ref{T1}, the expected number of misclassified nodes will always be greater than one, regardless of the chosen estimator. Specifically, given an estimator $\phi$, we denote the number of nodes that are misclassified by the estimator as $\varepsilon_{\phi}(n)$, and its expectation as $\mathbb{E}[\varepsilon_{\phi}(n)]$ (where the expectation is over the generation of the graphs as well as the intrinsic randomness in the estimator).

\begin{theorem}\label{T2}
When the model parameters of the MVSBM$(n,p,q,D)$ satisfy
\begin{equation}
    \lim_{n \to \infty}\frac{nI(p,q)}{\log n} < 2, \label{eq:T2}
    \end{equation}
then for any estimator $\phi$, the expected number of misclassified nodes must satisfy $\mathbb{E}[\varepsilon_{\phi}(n)] > 1$.
   %\ye{then exact recovery is impossible.} 
\end{theorem}

% \textbf{Theorem 2:} \textit{Under the assumption \(n \cdot \bar{p} = \omega(1)\), when the model parameters of the MVSBM$(n,p,q,D)$ satisfy}
% % \[
% % \lim_{n \rightarrow \infty} \inf \frac{nD(p(d)\|q(d))}{\log(n)} < 1,
% % \

Theorem~\ref{T2} establishes a condition (i.e., when the R\'{e}nyi divergence $I(p,q)$ is not large enough) under which it is impossible for any estimator to misclassify less than (or equal to) one node on average. The proof of Theorem~\ref{T2} is provided in Section~\ref{sec:proof_thm2}.

 %Essentially, if the R\'{e}nyi divergence $I(p,q)$ between the within-community and between-community probability distributions is not large enough, then exact recovery is not possible.

Combining Theorems~\ref{T1} and~\ref{T2} together, we conclude that there exists a \emph{sharp threshold} for community detection in the MVSBM. This threshold, denoted as
\begin{align}
\lim_{n \to \infty}\frac{nI(p,q)}{\log n} = 2, \label{eq:thre}
\end{align}
signifies that above the threshold, it is possible to ensure the number of misclassified nodes to be zero with high probability (i.e., achieve exact recovery of communities). Conversely, below the threshold, it becomes impossible to ensure the expected number of misclassified nodes to less than (or equal to) one, regardless of the chosen estimator.

\subsection{Corollaries derived from Theorems~\ref{T1} and~\ref{T2}}
\label{sec:Corollaries}

In the following, we apply our main results to the special cases of the MVSBM discussed in Section~\ref{sec:special}. This helps us to understand the information-theoretic limits of community detection in these special cases and also enables a comparison of our results with those from previous works.

\begin{corollary}[Special case 1] \label{cor1} For the standard SBM with two balanced communities (i.e., when $D =1$ and when $p$ and $q$ are Bernoulli distributions), we have 
$ I(p,q) = [p(1)^{1/2} - q(1)^{1/2} ]^2+ O(p(1)q(1)$,
and thus the threshold in~\eqref{eq:thre} becomes 
\begin{align}
\lim_{n \to \infty}\frac{n [ p(1)^{1/2} - q(1)^{1/2} ]^2}{\log n} = 2. 
\end{align}
\end{corollary}

We point out that the threshold in Corollary~\ref{cor1} matches the information-theoretic limit of the exact recovery of communities in the standard SBM~\cite{abbe2015exact}. Strictly speaking,  Corollary~\ref{cor1} also slightly generalizes the results in~\cite{abbe2015exact} as they only considered the parameter regime where $p(1)$ and $q(1)$ scale as $\Theta((\log n)/n)$, while our result allows for a much broader range for $p(1)$ and $q(1)$. 

\begin{corollary}[Special case 2] \label{cor2}For the extreme case when the $D$ graphs are identical (i.e., when the distributions $p$ and $q$ are chosen to satisfy~\eqref{eq:sp1} and~\eqref{eq:sp2}  respectively), we have $I(p,q) = [\alpha^{1/2} - \beta^{1/2}]^2 + O(\alpha \beta)$, and thus the threshold in~\eqref{eq:thre} becomes 
\begin{align}
\lim_{n \to \infty}\frac{n [\alpha^{1/2} - \beta^{1/2} ]^2}{\log n} = 2. 
\end{align}
\end{corollary}
In fact, the result in Corollary 2 is the same as that in Corollary~\ref{cor1}. This equivalence makes intuitive sense because observing multiple identical graphs is equivalent to observing only one graph in Special Case 2.

\begin{corollary}[Special case 3]\label{cor3}
For the other extreme case of observing $D$ independent graphs (i.e., when the distributions $p(\bm{d}) = \prod_{i=1}^D p_i(d_i)$ and $q(\bm{d}) = \prod_{i=1}^D q_i(d_i)$ are product distributions), we have $I(p,q) = \sum_{i=1}^D  [p_i(1)^{1/2} - q_i(1)^{1/2}]^2 + O\left(\max_{i\in [D]} (p_i(1)+ q_i(1))^2 \right).$
Thus the threshold in~\eqref{eq:thre} becomes 
\begin{align}
\lim_{n \to \infty}\frac{n \sum_{i=1}^D [p_i(1)^{1/2} - q_i(1)^{1/2}]^2}{\log n} = 2. 
\end{align}

\end{corollary}

The detailed proofs for Corollaries~\ref{cor1}-\ref{cor3} are provided in Section~\ref{sec:proof_Corollaries}.

\section{Proof of Theorem~\ref{T1}} \label{sec:proof_thm1}
In this section, we prove that the ML estimator achieves exact recovery if the condition in Theorem~\ref{T1} is satisfied. For any $\bm{x}' \in \mathcal{\bm{X}}$, the negative log-likelihood of $\bm{x}'$ with respect to the observed adjacency tensor $\bm{A}$ is defined as 
\begin{align}
L(\bm{A}|\bm{x}') := -\log \mathbb{P}_{ \bm{x}'}(\bm{A} ), 
\end{align}
where the subscript $\bm{x}'$ in $\mathbb{P}_{\bm{x}'}(\cdot)$ means that the tensor $\bm{A}$ is generated according to the communities determined by $\bm{x}'$.
The decision rule of the  ML estimator is then given  as
\begin{align}
\phi_{\text{ML}}(\bm{A}) = \underset{\bm{x}' \in \mathcal{\bm{X}} }{\text{argmin} } \ L \left(\bm{A}|\bm{x}' \right). \label{eq:ml}
\end{align}

We now analyze the performance of the ML estimator. We denote the \emph{error event} that the ML estimator fails to output the ground truth communities by \(E\). Our objective is to establish an upper bound on the probability of this error event, which is denoted as $\mathbb{P}(E)$. Note that 
\begin{align}
&\mathbb{P}(E) \notag \\
&=  \underset{(\bm{X} , \bm{A}) \sim \text{MVSBM}(n, p, q, D)  }{\mathbb{P}} \!\big(\phi_{\text{ML}}(\bm{A}) \!\ne\! \bm{X} \text{ and } \phi(\bm{A}) \!\ne\! -\bm{X} \big) \notag \\
&= \sum_{\bm{x} \in \mathcal{\bm{X}}}  \frac{1}{|\mathcal{\bm{X}}|} \mathbb{P}_{\bm{x}}\big(\phi_{\text{ML}}(\bm{A}) \!\ne\! \bm{x} \text{ and } \phi(\bm{A}) \!\ne\! -\bm{x} \big),
\end{align}
where $\bm{X}$ is the random variable chosen uniformly at random from the set $\mathcal{\bm{X}}$, and $\bm{A}$ is generated according to the rule given in Eqn.~\eqref{eq:A}. It is worth noting that, by symmetry, the error probabilities $\mathbb{P}_{\bm{x}}\big(\phi_{\text{ML}}(\bm{A}) \!\ne\! \bm{x} \text{ and } \phi(\bm{A}) \!\ne\! -\bm{x} \big)$ with respect to different realizations of $\bm{x} \in \mathcal{\bm{X}}$ are the same. Thus, without loss of generality, we only need to consider a specific $\bm{x}_0 \in \mathcal{X}$. For concreteness, we let $\bm{x}_0 = \{+1, \ldots, +1, -1, \ldots, -1\}$ be the vector such that the first $n/2$ elements are $+1$ and the remaining elements are $-1$ (as illustrated in Fig.~\ref{fig:label}). Thus, we have  
\begin{align}
\mathbb{P}(E) = \mathbb{P}_{\bm{x}_0} \big(\phi_{\text{ML}}(\bm{A}) \ne \bm{x}_0 \text{ and } \phi(\bm{A}) \ne -\bm{x}_0  \big).  \label{eq:12}
\end{align}
For any $\bm{x}' \in \{+1,-1\}^n$, we first define 
\begin{align}
\text{mis}(\bm{x}', \bm{x}_0) := \sum_{i=1}^n \mathbf{1}\{x'(i) \ne x_0(i) \}
\end{align}
as the total number of  different elements between \(\bm{x}'\) and \(\bm{x}_0\). Note that when $\phi_{\text{ML}}(\bm{A}) \ne \bm{x}_0$ and $\phi_{\text{ML}}(\bm{A}) \ne -\bm{x}_0$, we must have $\text{mis}(\phi_{\text{ML}}(\bm{A}), \bm{x}_0) = 2k$ for a certain $k \in [1: (n/2)-1]$. Thus, one can then reformulate Eqn.~\eqref{eq:12} as
\begin{align}
&\mathbb{P}_{\bm{x}_0} \big(\phi_{\text{ML}}(\bm{A}) \ne \bm{x}_0 \text{ and } \phi_{\text{ML}}(\bm{A}) \ne -\bm{x}_0  \big) \\
&=\mathbb{P}_{\bm{x}_0} \big( \cup_{k=1}^{(n/2)-1} \{ \text{mis}(\phi_{\text{ML}}(\bm{A}), \bm{x}_0) = 2k \} \big) \\
&\le \sum_{k=1}^{(n/2)-1}  \mathbb{P}_{\bm{x}_0} \big(  \text{mis}(\phi_{\text{ML}}(\bm{A}), \bm{x}_0) = 2k \big) \\
&= \sum_{k=1}^{(n/2)-1}  \ \sum_{\bm{x}' \in \mathcal{\bm{X}}: \text{mis}(\bm{x}',\bm{x}_0) = 2k } \mathbb{P}_{\bm{x}_0} (\phi_{\text{ML}}(\bm{A}) = \bm{x}') \\
&\le \sum_{k=1}^{(n/2)-1}  \ \sum_{\bm{x}' \in \mathcal{\bm{X}}: \text{mis}(\bm{x}',\bm{x}_0) = 2k } \mathbb{P}_{\bm{x}_0} (L(\bm{A}|\bm{x}') \le L(\bm{A}|\bm{x}_0)), \label{eq:17}
\end{align}
where Eqn.~\eqref{eq:17} follows from the estimation rule of the ML estimator introduced in Eqn.~\eqref{eq:ml}.

Next, we consider a fixed $k \in [1: (n/2)-1]$. A key observation is that for every $\bm{x}' \in \mathcal{\bm{X}}$ such that $\text{mis}(\bm{x}',\bm{x}_0) = 2k$, the induced error probability $\mathbb{P}_{\bm{x}_0} (L(\bm{A}|\bm{x}') \le L(\bm{A}|\bm{x}_0))$ is the same. Thus, it suffices to calculate this error probability by focusing on a specific $\bm{x}_k$ that differs from $\bm{x}_0$ in the index sets \([1: k]\) and \(\left[\frac{n}{2}+1: \frac{n}{2}+k\right]\). We illustrate $\bm{x}_k$ in Fig.~\ref{fig:label}.  ~Eqn.~\eqref{eq:17} is then equal to 
\begin{align}
\sum_{k=1}^{(n/2)-1} \!\!\! \mathbb{P}_{\bm{x}_0} (L(\bm{A}|\bm{x}_k) \le L(\bm{A}|\bm{x}_0)) \cdot |\bm{x}' \!\in\! \mathcal{\bm{X}}: \text{mis}(\bm{x}',\bm{x}_0) \!=\! 2k|. \label{eq:new}
\end{align}
Lemma~\ref{lemma:err} below presents an upper bound on $\mathbb{P}_{\bm{x}_0} (L(\bm{A}|\bm{x}_k) \le L(\bm{A}|\bm{x}_0))$, the probability that $\bm{x}_k$ is more likely than $\bm{x}_0$ (when $\bm{A}$ is generated according to the ground truth $\bm{x}_0$).

\begin{lemma} \label{lemma:err}
For any $k \in [1: (n/2)-1]$, we have
\begin{align}
\mathbb{P}_{\bm{x}_0} (L(\bm{A}|\bm{x}_k) \le L(\bm{A}|\bm{x}_0)) \le \exp\left\{-k(n-2k) \cdot I(p, q) \right\}. \notag  
\end{align}
In fact, the above inequality holds for any $\bm{x}' \in \mathcal{\bm{X}}$ such that $\text{mis}(\bm{x}',\bm{x}_0) = 2k$.
\end{lemma}

\begin{figure}[t] \centering 
\includegraphics[width=8.04cm,height=3.5cm]{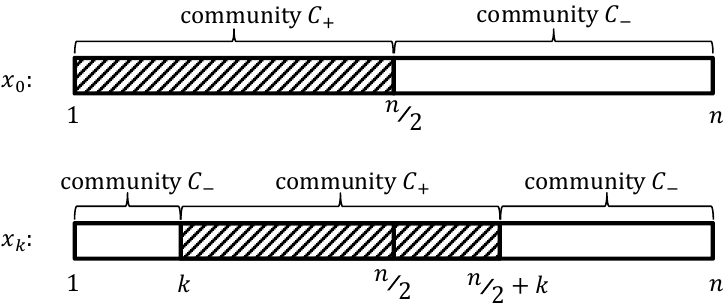} \caption{Illustration of the specific community labels $\bm{x}_0$ and $\bm{x}_k$, where $\bm{x}_0 = \{+1,\ldots, +1, -1, \ldots, -1  \}$. Note that $\bm{x}_k$ differs from $\bm{x}_0$ in the index sets \([1: k]\) and \(\left[\frac{n}{2}+1: \frac{n}{2}+k\right]\), thus $\text{mis}(\bm{x}_k,\bm{x}_0) = 2k$.  } \label{fig:label} 
\end{figure}

\begin{proof}[Proof of Lemma~\ref{lemma:err}]
Based on $\bm{x}_0$ and $\bm{x}_k$, we categorize the node pairs $(i,j)$ into five distinct classes, described as follows:

\begin{enumerate}
    \item For node pairs \((i, j)\) such that \(i \in [1: k]\) and \(j \in [k+1: \frac{n}{2}]\), it is clear from Fig.~\ref{fig:label} that nodes $i$ and $j$ belong to the same community under \(\bm{x}_0\), and belong to different communities under \(\bm{x}_k\). Thus, we have \(A_{ij} \sim p\) under \(\bm{x}_0\) and \(A_{ij} \sim q\) under \(\bm{x}_k\). The collection of such node pairs is denoted as $\mathcal{B}_1 \triangleq \{(i,j): i \in [1:k], j \in [k+1: \frac{n}{2}]\}.$
    \item For node pairs \((i, j)\) such that \(i \in [1: k]\) and \(j \in [\frac{n}{2} + k + 1: n]\), it is clear from Fig.~\ref{fig:label} that \(A_{ij} \sim q\) under \(\bm{x}_0\) and \(A_{ij} \sim p\) under \(\bm{x}_k\). The collection of such node pairs is denoted as $\mathcal{B}_2 \triangleq \{(i,j): i \in [1:k], j \in [\frac{n}{2} + k + 1: n]\}. $
    \item For node pairs \((i, j)\) such that \(i \in [\frac{n}{2}+1: \frac{n}{2}+k]\) and \(j \in [k+1: \frac{n}{2}]\), we have \(A_{ij} \sim q\) under \(\bm{x}_0\) and \(A_{ij} \sim p\) under \(\bm{x}_k\). The collection of such node pairs is denoted as $\mathcal{B}_3 \triangleq \{(i,j): i \in [\frac{n}{2}+1: \frac{n}{2}+k], j \in [k+1: \frac{n}{2}]\}. $
    \item For node pairs \((i, j)\) such that \(i \in [\frac{n}{2}+1: \frac{n}{2}+k]\) and \(j \in [\frac{n}{2}+k+1: n]\), we have \(A_{ij} \sim p\) under \(\bm{x}_0\) and \(A_{ij} \sim q\) under \(\bm{x}_k\). The collection of such node pairs is denoted as $\mathcal{B}_4 \triangleq \{(i,j): i \in [\frac{n}{2}+1: \frac{n}{2}+k], j \in [\frac{n}{2}+k+1: n] \}. $
    \item For node pairs \((i, j) \notin (\mathcal{B}_1\cup \mathcal{B}_2 \cup \mathcal{B}_3\cup \mathcal{B}_4) \), the distributions of $A_{ij}$ under both \(\bm{x}_0\) and \(\bm{x}_k\) are the same. The collection of such node pairs is denoted as $\mathcal{I}_{\checkmark} := \{(i, j) \in [n] \times [n] : i \ne j\} \setminus  (\mathcal{B}_1\cup \mathcal{B}_2 \cup \mathcal{B}_3\cup \mathcal{B}_4)$.
\end{enumerate}
For simplicity, we abbreviate the union of $\mathcal{B}_1, \mathcal{B}_2, \mathcal{B}_3, \mathcal{B}_4$ as 
$\mathcal{I}_{\times} := \mathcal{B}_1\cup \mathcal{B}_2 \cup \mathcal{B}_3\cup \mathcal{B}_4, $
where the subscript `$\times$' means that the node pairs in this set are classified \emph{incorrectly} under $\bm{x}_k$. Similarly, the subscript `$\checkmark$' associated with $\mathcal{I}_{\checkmark}$ means that the node pairs therein are classified \emph{correctly} under $\bm{x}_k$, and thus these node pairs can be ignored when calculating the error probability term $\mathbb{P}_{\bm{x}_0} (L(\bm{A}|\bm{x}_k) \le L(\bm{A}|\bm{x}_0))$.

\begin{comment}
\begin{itemize}
    \item $\mathcal{I}_{\times} := \{(i, j) : (i, j) \in (\mathcal{B}_1\cup \mathcal{B}_2 \cup \mathcal{B}_3\cup \mathcal{B}_4 )\} $ -- the set containing all errors.
    \item $\mathcal{I}_{\checkmark} := \{(i, j) : 1\le i < j \le n,  (i, j) \notin (\mathcal{B}_1\cup \mathcal{B}_2 \cup \mathcal{B}_3\cup \mathcal{B}_4) \} $ -- the set of correct classification.
\end{itemize}
\end{comment}

For ease of analyses, we decompose the adjacency tensor $\bm{A}$ into two sub-tensors denoted by $\bm{A}_{\times}$ and $\bm{A}_{\checkmark}$, where $\bm{A}_{\times} := \{\bm{A}_{ij}: (i,j) \in \mathcal{I}_{\times} \}$ consists of the elements of $\bm{A}$ whose indices belong to $\mathcal{I}_{\times}$, and $\bm{A}_{\checkmark} := \{\bm{A}_{ij}: (i,j) \in \mathcal{I}_{\checkmark} \}$ consists of the elements of $\bm{A}$ whose indices belong to $\mathcal{I}_{\checkmark}$. Let $\mathcal{A}$ be the set of all possible realizations of the adjacency tensor $\bm{A}$. Let $\mathcal{A}_{\times}$ (resp. $\mathcal{A}_{\checkmark}$) be the set of all possible realizations of the sub-tensor $\bm{A}_{\times}$ (resp. $\bm{A}_{\checkmark}$). 

Then, we have
\begin{align}
&\mathbb{P}_{\bm{x}_0} (L(\bm{A}|\bm{x}_k) \le L(\bm{A}|\bm{x}_0)) \quad \quad \quad \quad \quad \quad \quad \quad \quad \quad \quad \quad \notag  \\
& = \sum_{a \in \mathcal{A}} \mathbb{P}_{\bm{x}_0}(\bm{A} = a) \cdot \bm{1}\left\{\frac{L(a|\bm{x}_k)}{L(a|\bm{x}_0)} \le 1 \right\} \notag \\
&= \sum_{a' \in \mathcal{A}_{\times}} \mathbb{P}_{\bm{x}_0}(\bm{A}_{\times} \!=\! a') 
\! \sum_{\tilde{a} \in \mathcal{A}_{\checkmark}} \mathbb{P}_{\bm{x}_0}(\bm{A}_{\checkmark} = \tilde{a}) \notag \\
&\qquad\qquad \cdot \bm{1} \left\{ \frac{\mathbb{P}_{\bm{x}_k}(\bm{A}_{\times} = a')\cdot \mathbb{P}_{\bm{x}_k}(\bm{A}_{\checkmark} = \tilde{a})}{\mathbb{P}_{\bm{x}_0}(\bm{A}_{\times} = a')\cdot \mathbb{P}_{\bm{x}_0}(\bm{A}_{\checkmark} = \tilde{a})} \ge 1 \right\} \label{eq:20}\\
&= \sum_{a' \in \mathcal{A}_{\times}} \mathbb{P}_{\bm{x}_0}(\bm{A}_{\times} = a') \sum_{\tilde{a} \in \mathcal{A}_{\checkmark}} \mathbb{P}_{\bm{x}_0}(\bm{A}_{\checkmark} = \tilde{a}) \notag \\
&\qquad\qquad  \cdot \bm{1} \{\mathbb{P}_{\bm{x}_k}(\bm{A}_{\times} = a') \ge \mathbb{P}_{\bm{x}_0}(\bm{A}_{\times} = a') \}  \label{eq:21}\\
&= \sum_{a' \in \mathcal{A}_{\times}} \mathbb{P}_{\bm{x}_0}(\bm{A}_{\times} = a') \cdot \bm{1} \left\{ \frac{\mathbb{P}_{\bm{x}_k}(\bm{A}_{\times} = a')}{\mathbb{P}_{\bm{x}_0}(\bm{A}_{\times} = a')} \ge 1 \right\} \notag  \\
&\le \sum_{a' \in \mathcal{A}_{\times}}  \min\{\mathbb{P}_{\bm{x}_k}(\bm{A}_{\times} = a'), \mathbb{P}_{\bm{x}_0}(\bm{A}_{\times} = a') \} \label{eq:23}
\\
&\le \sum_{a' \in \mathcal{A}_{\times}} \mathbb{P}_{\bm{x}_k}(\bm{A}_{\times} = a')^t \cdot \mathbb{P}_{\bm{x}_0}(\bm{A}_{\times} = a')^{1-t}, \label{eq:24}
\end{align}
where Eqn.~\eqref{eq:20} follows from the definition of the negative log-likelihood function and the fact that $\bm{A}_{\times}$ and $\bm{A}_{\checkmark}$ are generated independently, Eqn.~\eqref{eq:21} holds since $\mathbb{P}_{\bm{x}_0}(\bm{A}_{\checkmark} = \tilde{a}) = \mathbb{P}_{\bm{x}_k}(\bm{A}_{\checkmark} = \tilde{a})$ by the definitions of $\bm{A}_{\checkmark}$ and $\mathcal{I}_{\times}$. Eqn.~\eqref{eq:23} follows from the fact that for any constants $b,c>0$, we have $b\cdot \bm{1}\{c \ge b\} \le  \min\{b,c\}$. Eqn.~\eqref{eq:24} holds for any $t \in [0,1]$, and is due to the fact that $\min\{b,c\} \le b^t \cdot c^{1-t}$ for any constants $t \in [0,1]$ and  $b,c > 0$.
 
% There are four different situations concerning the misclassified points based on their relationships among the communities. The number of events for each situation is denoted by \(T = |\mathcal{B}_1| = |\mathcal{B}_2| =|\mathcal{B}_3|=|\mathcal{B}_4|= k \cdot \left(\frac{n}{2} - k\right)\). For a specific event $c \in \{\mathcal{B}_1,\mathcal{B}_2, \mathcal{B}_3, \mathcal{B}_4 \} $, the connection vectors can be represented as \({a}_T^{c} \in \{0,1\}^D\).

Note that the sub-tensor $\bm{A}_{\times}$ can further be decomposed into the connection vectors $\bm{A}_{ij}$ for node pairs  $(i,j) 
\in \mathcal{I}_{\times}$, where $\mathcal{I}_{\times} = \mathcal{B}_1\cup \mathcal{B}_2 \cup \mathcal{B}_3\cup \mathcal{B}_4$ by definition. One can check that, for each $u \in \{1,2,3,4\}$, the cardinality of the set $\mathcal{B}_u$ is $k\cdot (\frac{n}{2}-k)$.  For any node pair $(i,j) \in \mathcal{B}_u$, the distribution of $\bm{A}_{ij}$ under $\bm{x}_0$ and the distribution of $\bm{A}_{ij}$ under $\bm{x}_k$ are known (by recalling the definition of $\mathcal{B}_u$ described above). Therefore, we can reformulate Eqn.~\eqref{eq:24} as

\begin{align}
&\sum_{a' \in \mathcal{A}_{\times}} \mathbb{P}_{\bm{x}_k}(\bm{A}_{\times} = a')^t \cdot \mathbb{P}_{\bm{x}_0}(\bm{A}_{\times} = a')^{1-t} \notag \\
&=\bigg(\sum_{\bm{d} \in \{0,1 \}^D} q(\bm{d})^t p(\bm{d})^{1-t} \bigg)^{|\mathcal{B}_1|+|\mathcal{B}_4|} \notag \\
&\qquad\qquad\qquad \cdot \bigg(\sum_{\bm{d} \in \{0,1 \}^D} p(\bm{d})^t q(\bm{d})^{1-t} \bigg)^{|\mathcal{B}_2|+|\mathcal{B}_3|} \notag \\
&=\exp\Bigg\{(|\mathcal{B}_1|+|\mathcal{B}_4|)\log\bigg(\sum_{\bm{d} \in \{0,1 \}^D} q(\bm{d})^t p(\bm{d})^{1-t}\bigg) \notag \\
&\qquad\qquad +(|\mathcal{B}_2|+|\mathcal{B}_3|)\log\bigg(\sum_{\bm{d} \in \{0,1 \}^D} p(\bm{d})^t q(\bm{d})^{1-t}\bigg) \Bigg\} \notag 
\end{align}
\begin{align}
&= \exp\left\{-k(n-2k) \cdot D_t(p \| q) \right\},\qquad\qquad\qquad \qquad   \label{eq:guo}
\end{align}
where $D_t(p \| q) := - \log (\sum_{\bm{d} \in \{0,1\}^D} p(\bm{d})^t q(\bm{d})^{1-t}) - \log(\sum_{\bm{d} \in \{0,1\}^D} q(\bm{d})^t p(\bm{d})^{1-t})$. Note that the above derivations are valid for \emph{any} $t\in [0,1]$, thus we can choose the value of $t$ that maximizes $D_t(p \| q)$. In fact, one can show that the maximizer of $D_t(p \| q)$ is  $t= 1/2$, and $D_{1/2}(p \| q)$ is exactly equal to $I(p,q)$.  Combining Eqns.~\eqref{eq:24} and~\eqref{eq:guo} together, we have 
\begin{align}
\mathbb{P}_{\bm{x}_0} (L(\bm{A}|\bm{x}_k) \le L(\bm{A}|\bm{x}_0)) \le \exp\left\{-k(n-2k) \cdot I(p, q) \right\}. \notag
\end{align}
This completes the proof of Lemma~\ref{lemma:err}.

\end{proof}

Referring to Eqns.~\eqref{eq:17}-\eqref{eq:new}, we observe that it is also necessary to compute the number of $\bm{x}' \in \mathcal{\bm{X}}$ satisfying $\text{mis}(\bm{x}',\bm{x}_0) = 2k$.  Using the inequality $\binom{n}{k} \le (en/k)^k$ (where $e$ is the Euler's number), we have that for $k \in [1:n/4]$,
\begin{equation}
\begin{aligned}
\left|\left\{\bm{x}': \operatorname{mis}(\bm{x}', \bm{x}_0)=2k\right\}\right| &= \binom{\frac{n}{2}}{k} \cdot \binom{\frac{n}{2}}{k} \\
&\leq \left(\frac{e \frac{n}{2}}{k}\right)^{k} \cdot \left(\frac{e \frac{n}{2}}{k}\right)^{k} \\
&= \exp\left(2k\left(1 + \log \frac{n}{2k} \right)\right). \label{eq:26}
\end{aligned}
\end{equation}
For $k \in [n/4:n/2-1]$, we have $$\left|\left\{\bm{x}' \!: \operatorname{mis}(\bm{x}', \bm{x}_0) \!=\! 2k\right\}\right| = \left|\left\{\bm{x}' \!: \operatorname{mis}(\bm{x}', \bm{x}_0) \!=\! 2\left(\frac{n}{2} - k \right)\right\}\right|,$$
which  implies that 
\begin{align}
&\sum_{k=1}^{(n/2)-1} \left|\left\{\bm{x}': \operatorname{mis}(\bm{x}', \bm{x}_0)=2k\right\}\right| \cdot \exp\left\{ -k(n-2k) I(p, q)  \right\} \notag  \\
&\le 2 \sum_{k=1}^{n/4} \left|\left\{\bm{x}': \operatorname{mis}(\bm{x}', \bm{x}_0)=2k\right\}\right| \cdot \exp\left\{ -k(n-2k) I(p, q)  \right\}. \label{eq:sym}
\end{align}

% Note that for sufficiently large \( n \) and $\epsilon > 0$, if the model parameters satisfy,
% \[I > 2(1 + \epsilon) \frac{\log n}{n}, \] 
Substituting Lemma~\ref{lemma:err} and Eqns.~\eqref{eq:26}-\eqref{eq:sym} into~\eqref{eq:new}, we can upper-bound the error probability $\mathbb{P}(E)$ as
\begin{align}
\mathbb{P}(E) \leq 2\sum_{k=1}^{n/4} \exp\left\{ 2k\left(1+ \log \frac{n}{2k} \right) - 2k\left(\frac{n}{2} - k\right)I(p,q) \right\}. \notag
\end{align}
When the model parameters of the MVSBM$(n,p,q,D)$ satisfy    $\lim_{n \to \infty}\frac{nI(p,q)}{\log n} > 2$, we know  that there exists a constant $\epsilon > 0$ such that $I(p,q) \ge 2(1+\epsilon)\frac{\log n}{n}$. Therefore,

% \begin{comment}
% \textbf{Lemma 1:} \textit{In a multi-view SBM with two communities of equal size and with connection probabilities governed by the model, the probability of event \(E\) occurring is bounded as:}

% \begin{equation}
%     \begin{aligned}
%   \mathbb{P}(E) \leq &\sum_{k=1}^{\frac{n}{4}} \exp \left\{ 2k \left( \log \frac{n}{2k} + 1 \right) \right\} \\
%   &\quad\cdot \exp \left\{ -2k \cdot \left( \frac{n}{2} - k \right) I \right\}.
% \end{aligned}
% \end{equation}

% This lemma essentially argues that the probability of a failure in maximum likelihood estimation approaches 0 under the bound provided by the divergence \(I\).
% \end{comment}
\begin{align}
    \mathbb{P}(E) &\leq 2\sum_{k=1}^{n/4} \exp\left\{ 2k\left( 1 - \log 2k + \log n \right) \right\} \notag \\
&\quad\quad\cdot  \exp\left\{ k\left( -\left(\frac{1}{2} - \frac{k}{n}\right) \cdot 2n \cdot 2(1 + \epsilon) \frac{\log n}{n} \right) \right\} \notag \\ 
&\le 2\sum_{k=1}^{n/4} \exp\left\{ 2k \left( \frac{2k}{n} \log n-\log 2k \right) \right\} \notag \\
    &\quad\quad \cdot \exp\left\{ k \left( 2 -\left( \frac{1}{2} - \frac{k}{n} \right) 4\epsilon \log n   \right) \right\} \label{eq:exp1}
\end{align}

\begin{align}
&\le 2\sum_{k=1}^{n/4} \exp\left\{ 2k \left(1 -\log (2k) + \frac{2k}{n} \log n - \frac{\epsilon}{2}\log n  \right) \right\} \label{eq:exp2} \\
    &= 2n^{-\epsilon k} \cdot \sum_{k=1}^{n/4} \exp\left\{ -2k \left( \log (2k) - \frac{2k}{n} \log n - 1 \right) \right\} \notag \\
&\leq 2n^{-\epsilon} \cdot \sum_{k=1}^{n/4} \exp\left\{ 2k \left( 1 - \frac{1}{3} \log (2k)  \right) \right\}, \label{eq:P_A_E_3}
\end{align}
where Eqn.~\eqref{eq:exp1} holds since \( 2 \log n - 2(1+\epsilon) \log n \leq 0 \), Eqn.~\eqref{eq:exp2} holds since \( \left( \frac{1}{2} - \frac{k}{n} \right) 4 \ge 1 \) when \( 1 \leq k \leq \frac{n}{4} \). Eqn.~\eqref{eq:P_A_E_3} follows from the facts that \( n^{-\epsilon k} \leq n^{-\epsilon} \) and \( \log (2k) - \frac{2k}{n} \log n \geq \frac{1}{3} \log (2k) \).

By noting that there exists a universal constant $C_0 > 0$ such that the term $ 2 \sum_{k=1}^{n/4} \exp\left\{ 2k \left(1 - \frac{1}{3} \log 2k \right) \right\}$ can be upper-bounded by the constant $C_0$, we eventually obtain that

%Where the sum \( \sum_{k=1}^{\frac{n}{4}} \exp\left\{ -2k \left( \frac{1}{3} \log 2k - 1 \right) \right\} \) is a constant \( C \) for sufficiently large \( k \). Thus, we conclude:

\begin{equation}
    \mathbb{P}(E) \leq n^{-\epsilon} \cdot C_0, \label{eq:P_A_E_conclusion}
\end{equation}
implying that \( \mathbb{P}(E) \to 0 \) as \( n \to \infty \). It means that the ML estimator $\phi_{\text{ML}}$ achieves  exact recovery when the condition $\lim_{n \to \infty}\frac{nI(p,q)}{\log n} > 2$ in Theorem~\ref{T1} holds. This completes the proof of Theorem~\ref{T1}.

\section{Proof of Theorem~\ref{T2}} \label{sec:proof_thm2}
Recall that Theorem~\ref{T2} states that, as long as Eqn.~\eqref{eq:T2} holds, the expected number of misclassified nodes $\mathbb{E}[\varepsilon_{\phi}(n)]$ must be larger than one regardless of the estimator $\phi$ used. To prove Theorem~\ref{T2}, it is equivalent to showing that ``if there exists an algorithm $\phi$ satisfying $\mathbb{E}[\varepsilon_{\phi}(n)] \le 1$, then the parameters of MVSBM must satisfy $\lim_{n \rightarrow \infty} \frac{n I(p,q)}{\log(n)} \ge 2.$''

\begin{comment}
To prove Theorem 2, we focus on establishing its contrapositive statement, because it becomes impossible to ensure the expected number of misclassified nodes to be less than (or equal to) one. Under the assumptions $A_1$ and $A_2$, if there exists an algorithm with a number of misclassified points \( k < 1 \), then the parameters of MVSBM must satisfy
\[
\lim_{n \rightarrow \infty} \inf \frac{n I(p,q)}{\log(n)} > 2.
\]
\end{comment}

The proof for Theorem~\ref{T2} uses a change-of-measure technique, which is developed by the prior work~\cite{yun2016optimal} for investigating community detection in the labeled SBM. Let \(\Phi\) be an abbreviation of the true model parameters for generating the node labels $\bm{X}$ and the adjacency tensor $\bm{A}$, such that $$\mathbb{P}_{\Phi}(\cdot) = \underset{(\bm{X}, \bm{A}) \sim \text{MVSBM}(n, p, q, D)  }{\mathbb{P}} (\cdot),$$ 
and let \(\mathbb{E}_{\Phi}[\cdot]\) be the corresponding expectation. Thus, for any estimator $\phi$, its expected number of misclassified nodes $\mathbb{E}[\varepsilon_{\phi}(n)]$ can also be rewritten as $\mathbb{E}_{\Phi}[\varepsilon_{\phi}(n)]$.  A smaller value of \(\mathbb{E}_{\Phi}[\varepsilon_{\phi}(n)]\) indicates a better community recovery performance, as it represents smaller misclassification nodes in expectation.

Let $\mathcal{N} \subseteq [n]$ be the set of nodes that are misclassified. For any node $v \in [n]$, due to the symmetrical structure of the MVSBM, it is clear that 
\begin{align}
\mathbb{E}_{\Phi}[\varepsilon_{\phi}(n)] = n \cdot \mathbb{P}_{\Phi}(v \in \mathcal{N}).
\end{align}

To assist our analyses, we introduce another model $\Psi$ that is different from but correlated with the true model $\Phi$. To be specific, 
\begin{enumerate}
    \item The node labels $\bm{X}=(X(1), X(2), \ldots, X(n))$ generated under $\Psi$ are exactly the same as the node labels generated under $\Phi$.
    \item Under $\Psi$, for any node $v \in [2:n]$, the connection vector $\bm{A}_{1v}$ between nodes $1$ and $v$ follows from distribution $\tilde{p} \in \Delta(\{0,1\}^D)$ if $X(v) = +1$, and $\tilde{q}\in \Delta(\{0,1\}^D)$ if $X(v) = -1$. Here, the two distributions $\tilde{p}$ and $\tilde{q}$ are chosen to satisfy 
    \begin{align}
    D_{\text{KL}}(\tilde{p}\Vert p) + D_{\text{KL}}(\tilde{q}\Vert q) = D_{\text{KL}}(\tilde{p}\Vert q)  + D_{\text{KL}}(\tilde{q}\Vert p), \label{eq:LR}   
    \end{align}
where $D_{\text{KL}}(\cdot \Vert \cdot) $ represents the KL-divergence between two distributions.
    The existence of such $\tilde{p}$ and $\tilde{q}$ can be proved in an analogous fashion as in~\cite[Lemma~7]{yun2016optimal}.

\item For any pair of nodes $(v,v')$ such that $v \ne 1$ and $v' \ne 1$, the connection vector $\bm{A}_{vv'}$ under $\Psi$ is the same as that under~$\Phi$.
\end{enumerate}

Given the adjacency tensor $\bm{A}$, we then introduce the \emph{log-likelihood ratio} $R$ of observing $\bm{A}$ with respect to the models $\Phi$ and $\Psi$:
\begin{align}
 R &:= \log \frac{\mathbb{P}_{\Psi}(\bm{A}) }{\mathbb{P}_{\Phi}(\bm{A})} = \sum_{v=2}^n \log \frac{\mathbb{P}_{\Psi}(\bm{A}_{1v}) }{\mathbb{P}_{\Phi}(\bm{A}_{1v})}. \label{eq:R}
\end{align}
A larger value of \(\mathbb{E}_{\Psi}[R]\) indicates a better fit of the observed graphs to the model \(\Psi\) compared to the true model $\Phi$. By the definitions of $\Phi$ and $\Psi$, one can obtain that:
\begin{itemize}
\item $\mathbb{P}_{\Psi}(\bm{A}_{1v}) =\tilde{p}(\bm{A}_{1v})$ and $\mathbb{P}_{\Phi}(\bm{A}_{1v}) = p(\bm{A}_{1v})$ if the node labels $X(1) = +1$ and $X(v) = +1$.
\item $\mathbb{P}_{\Psi}(\bm{A}_{1v}) =\tilde{p}(\bm{A}_{1v})$ and $\mathbb{P}_{\Phi}(\bm{A}_{1v}) = q(\bm{A}_{1v})$ if the node labels $X(1) = -1$ and $X(v) = +1$.
\item $\mathbb{P}_{\Psi}(\bm{A}_{1v}) = \tilde{q}(\bm{A}_{1v})$ and $\mathbb{P}_{\Phi}(\bm{A}_{1v}) = p(\bm{A}_{1v})$ if the node labels $X(1) = -1$ and $X(v) = -1$.
\item $\mathbb{P}_{\Psi}(\bm{A}_{1v}) = \tilde{q}(\bm{A}_{1v})$ and $\mathbb{P}_{\Phi}(\bm{A}_{1v}) = q(\bm{A}_{1v})$ if the node labels $X(1) = +1$ and $X(v) = -1$.
\end{itemize}

Following the change-of-measure technique developed in~\cite{yun2016optimal}, we relate the expected number of misclassified nodes $\mathbb{E}_{\Phi}[\varepsilon_{\phi}(n)]$ to $\mathbb{P}_{\Psi}(R \le f(n))$ for some function $f(n)$.

\begin{lemma} \label{lemma:change}
    For any function $f(n)$, we have
$$\mathbb{P}_{\Psi}(R \le f(n)) \le e^{f(n)}\frac{\mathbb{E}_{\Phi}[\varepsilon_{\phi}(n)]}{n} + \frac{1}{2}.$$
\end{lemma}
\begin{proof}[Proof of Lemma~\ref{lemma:change}]
The proof relies on a change-of-measure technique that  translates the measure from $\mathbb{P}_{\Psi}(\cdot)$ to $\mathbb{P}_{\Phi}(\cdot)$. We refer the readers to Section~\ref{sec:lemma2} for the detailed proofs.
\end{proof}

Specializing $f(n) = \log \left( \frac{n}{\mathbb{E}_{\Phi}\left[\varepsilon_{\tau}(n)\right]} \right) - 2\log 2$ yields that
\begin{equation}
\mathbb{P}_{\Psi}\left\{R \leq \log \left(\frac{n}{\mathbb{E}_{\Phi}\left[\varepsilon_{\tau}(n))\right]}\right) - 2\log 2\right\} \leq \frac{3}{4}.
\label{eq:P_phi6}
\end{equation}

 Next, we apply the Chebyshev's inequality to analyze \( \mathbb{P}_{\Psi}\left\{R \leq f(n)\right\} \), which yields that
\begin{equation}
\mathbb{P}_{\Psi}\left\{R \leq \mathbb{E}_{\Psi}[R] +  \sqrt{4\mathbb{E}_{\Psi}\left[\left(R - \mathbb{E}_{\Psi}[R]\right)^2\right]}\right\} \geq \frac{3}{4}.
\label{eq:P_phi7}
\end{equation}
Eqns.~\eqref{eq:P_phi6} an~\eqref{eq:P_phi7} together imply that 
\begin{align}
&\log \left( \frac{n}{\mathbb{E}_{\Phi}[\varepsilon_{\phi}(n)]} \right) - 2\log 2 \notag \\
&\qquad\qquad\qquad \leq \mathbb{E}_{\Psi}[R] + \sqrt{ 4\mathbb{E}_{\Psi}\left[\left(R - \mathbb{E}_{\Psi}[R]\right)^2\right]}.
\end{align}
Thus, for any estimator $\phi$, if  \( \mathbb{E}_{\Phi}\left[\bar{\varepsilon}_{\tau}(n)\right] \leq 1 \), the following condition must hold:
\begin{equation}
\log \left( n \right) - 2\log 2 \leq \mathbb{E}_{\Psi}[R] + \sqrt{ 4\mathbb{E}_{\Psi}\left[\left(R - \mathbb{E}_{\Psi}[R]\right)^2\right]}.
\label{eq:nece_con} 
\end{equation}
It then remains to evaluate $\mathbb{E}_{\Psi}[R]$ and $\mathbb{E}_{\Psi}[\left(R - \mathbb{E}_{\Psi}[R]\right)^2]$.

\subsubsection{Bounding $\mathbb{E}_{\Psi}[R]$}
Recall from Section~\ref{sec:proof_thm1} that we defined $\bm{x}_0 = \{+1, \ldots, +1, -1, \ldots, -1\} \in \mathcal{X}$ as a special node label vector. Let $-\bm{x}_0 = \{-1, \ldots, -1, +1, \ldots, +1\}$. Note that 
\begin{align}
\mathbb{E}_{\Psi}[R] &= \frac{1}{2} \mathbb{E}_{\Psi}[R|X(1) = +1] + \frac{1}{2} \mathbb{E}_{\Psi}[R|X(1) = -1]   \label{eq:sym} \\
&=\frac{1}{2} \mathbb{E}_{\Psi}[R|\bm{X} = \bm{x}_0] +\frac{1}{2}\mathbb{E}_{\Psi}[R|\bm{X} = -\bm{x}_0], \label{eq:sym2}
\end{align}
where~\eqref{eq:sym2} is due to the fact that  the expectation of $R$ conditioned on every $\bm{x} \in \{\bm{x}' \in \mathcal{X}: x'(1) = +1\}$ is the same, and the expectation of $R$ conditioned on every $\bm{x} \in \{\bm{x}' \in \mathcal{X}: x'(1) = -1\}$ is also the same. For the first term in~\eqref{eq:sym}, based on the definition of $R$ as well as the distributions of $\bm{A}_{1v}$ under both $\Psi$ and $\Phi$, we have
\begin{align}
&\mathbb{E}_{\Psi}[R|\bm{X} = \bm{x}_0] \label{eq:kai1}\\
&=\sum_{v=2}^{n/2}\mathbb{E}_{\tilde{p}}\left[\frac{\tilde{p}(\bm{A}_{1v})}{p(\bm{A}_{1v})} \right] + \sum_{v=n/2+1}^{n}\mathbb{E}_{\tilde{q}}\left[\frac{\tilde{q}(\bm{A}_{1v})}{q(\bm{A}_{1v})} \right], \label{eq:kai2}
\end{align}
which implies that
\begin{align}
&\frac{n-1}{2} \left[ D_{\text{KL}}(\tilde{p}\Vert p) + D_{\text{KL}}(\tilde{q}\Vert q) \right] \notag \\
 &\le   \mathbb{E}_{\Psi}[R|\bm{X} = \bm{x}_0] \le \frac{n}{2} \left[ D_{\text{KL}}(\tilde{p}\Vert p) + D_{\text{KL}}(\tilde{q}\Vert q) \right].
\end{align}
Similarly, for the second term, we have
\begin{align}
&\frac{n-1}{2} \left[ D_{\text{KL}}(\tilde{q}\Vert p) + D_{\text{KL}}(\tilde{p}\Vert q) \right] \notag \\
 &\le \mathbb{E}_{\Psi}[R|\bm{X} = -\bm{x}_0] \le  \frac{n}{2} \left[ D_{\text{KL}}(\tilde{q}\Vert p) + D_{\text{KL}}(\tilde{p}\Vert q) \right].   
\end{align}
By Eqn.~\eqref{eq:LR}, we eventually obtain that
\begin{align}
    &\frac{n-1}{2} \left[ D_{\text{KL}}(\tilde{p}\Vert p) + D_{\text{KL}}(\tilde{q}\Vert q) \right] \notag \\
    &\qquad\quad \le \mathbb{E}_{\Psi}[R] \le \frac{n}{2} \left[ D_{\text{KL}}(\tilde{p}\Vert p) + D_{\text{KL}}(\tilde{q}\Vert q) \right].
\end{align} 

\subsubsection{Upper-bounding $\mathbb{E}_{\Psi}[\left(R - \mathbb{E}_{\Psi}[R]\right)^2]$} First note that 
\begin{align}
    \mathbb{E}_{\Psi}[\left(R - \mathbb{E}_{\Psi}[R]\right)^2] = \mathbb{E}_{\Psi}[R^2] - (\mathbb{E}_{\Psi}[R])^2. 
\end{align}
Since $\mathbb{E}_{\Psi}(R)$ has been analyzed, it remains to analyze $\mathbb{E}_{\Psi}[R^2]$. Similar to~\eqref{eq:sym}-\eqref{eq:sym2}, one can show that
\begin{align}
\mathbb{E}_{\Psi}[R^2] =\frac{1}{2} \mathbb{E}_{\Psi}[R^2|\bm{X} = \bm{x}_0] +\frac{1}{2}\mathbb{E}_{\Psi}[R^2|\bm{X} = -\bm{x}_0].  \label{eq:46}
\end{align}
For the first term in~\eqref{eq:46}, we have
\begin{align}
&\mathbb{E}_{\Psi}[R^2|\bm{X}= \bm{x}_0]  \\
&= \mathbb{E}_{\Psi} \Bigg[\Bigg( \sum_{v=2}^{n/2} \log \frac{\tilde{p}(\bm{A}_{1v})}{p(\bm{A}_{1v})} + \sum_{v= \frac{n}{2}+1}^n \log \frac{\tilde{q}(\bm{A}_{1v})}{q(\bm{A}_{1v})} \Bigg)^2 \Bigg] \\
&= \sum_{v=2}^{n/2} \mathbb{E}_{\Psi} \left[\log \frac{\tilde{p}(\bm{A}_{1v})}{p(\bm{A}_{1v})} \right]^2 + \sum_{v=\frac{n}{2}+1}^{n} \mathbb{E}_{\Psi}\left[\log \frac{\tilde{q}(\bm{A}_{1v})}{q(\bm{A}_{1v})} \right]^2  \notag \\
&+\sum_{v,v'\in [2:\frac{n}{2}]:v\ne v'} \mathbb{E}_{\Psi} \left[\log \frac{\tilde{p}(\bm{A}_{1v})}{p(\bm{A}_{1v})}\right]  \cdot \mathbb{E}_{\Psi} \left[ \log \frac{\tilde{p}(\bm{A}_{1v'})} {p(\bm{A}_{1v'})} \right] \notag \\
&+\sum_{v,v'\in [\frac{n}{2}+1:n]:v\ne v'} \mathbb{E}_{\Psi} \left[\log \frac{\tilde{q}(\bm{A}_{1v})}{q(\bm{A}_{1v})}\right]  \cdot \mathbb{E}_{\Psi} \left[ \log \frac{\tilde{q}(\bm{A}_{1v'})} {q(\bm{A}_{1v'})} \right] \notag \\
&+2 \sum_{v \in [2:\frac{n}{2}]} \sum_{v'\in [\frac{n}{2}+1:n]} \mathbb{E}_{\Psi} \left[\log \frac{\tilde{p}(\bm{A}_{1v})}{p(\bm{A}_{1v})}\right]  \cdot \mathbb{E}_{\Psi} \left[ \log \frac{\tilde{q}(\bm{A}_{1v'})} {q(\bm{A}_{1v'})} \right] \notag \\
&\le \frac{n}{2} \log(\rho) D_{\text{KL}}(\tilde{p}\Vert p) + \frac{n}{2} \log(\rho) D_{\text{KL}}(\tilde{q}\Vert q) \notag \\
&\qquad\qquad\qquad +  \frac{n^2}{4} [D_{\text{KL}}(\tilde{p}\Vert p) + D_{\text{KL}}(\tilde{q}\Vert q) ]^2,  \label{eq:49}
\end{align}
where the last inequality is due to assumption (A1). Following a similar procedure, one can also upper-bound the second term of~\eqref{eq:46} by the RHS of~\eqref{eq:49}.   Therefore, we have 
\begin{align}
\mathbb{E}_{\Psi}[\left(R - \mathbb{E}_{\Psi}[R]\right)^2] &\le  \frac{n}{2} \log(\rho)[D_{\text{KL}}(\tilde{p}\Vert p) + D_{\text{KL}}(\tilde{q}\Vert q)] \notag \\
&+\frac{2n-1}{4}[D_{\text{KL}}(\tilde{p}\Vert p) + D_{\text{KL}}(\tilde{q}\Vert q) ]^2.
\end{align}

Moreover, it can be shown that 
    \begin{align}
    &\lim_{n \to \infty} \frac{D_{\text{KL}}(\tilde{p}\Vert p) + D_{\text{KL}}(\tilde{q}\Vert q)}{I(p,q)}= 1.
    \end{align}
Thus, by substituting the upper bounds of $\mathbb{E}_{\Psi}(R)$ and $\mathbb{E}_{\Psi}[\left(R - \mathbb{E}_{\Psi}[R]\right)^2]$ to Eqn.~\eqref{eq:nece_con} as well as considering sufficiently large $n$,  we obtain the necessary condition
\[
\lim_{n \rightarrow \infty} \frac{n I(p,q)}{\log(n)} > 2.
\]
This completes the proof of Theorem~\ref{T2}.

\section{Proof of Corollaries and Lemma~\ref{lemma:change}}
In this section, we prove the corollaries of the MVSBM addressed in Section ~\ref{sec:Corollaries}. And the detailed proof of Lemma~\ref{lemma:change} discussed in Section~\ref{sec:proof_thm2} is also provided.
\subsection{Proof of the Corollaries} \label{sec:proof_Corollaries}
\subsubsection{Proof of Corollary~\ref{cor1}}
We consider Special Case~1 in which only one graph is observed and the distributions $p$ and $q$ are both Bernoulli distributions. Note that by our assumption, both $p(1) = o(1)$ and $q(1) = o(1)$ hold.  In this case, the divergence $I(p,q)$ can be reformulated as

%Assuming \(p(1) = \frac{a \log n}{n}\) and \(q(1) = \frac{b \log n}{n}\), the model parameters satisfy \(|\sqrt{a} - \sqrt{b}| \geq \sqrt{2}\) for exact community recovery. We get:
\begin{align}
    &I(p,q) \notag\\
    &= -2 \log \left[ [p(0)q(0)]^{1/2} + [p(1)q(1)]^{1/2} \right]
    \label{eq:fine} \\
    &= -2 \log \left[ [(1-p(1))(1-q(1))]^{1/2} + [p(1)q(1)]^{1/2} \right] \label{eq:cor1} \\
    &= -2 \log \left[ [1 \!-\! p(1) \!-\! q(1) \!+\! p(1)q(1)]^{1/2} + [p(1)q(1)]^{1/2} \right]  \\
    &= -2 \log \left[1 - \frac{p(1)+q(1)}{2} + [p(1)q(1)]^{1/2} + O(p(1)q(1)) \right] \label{eq:taylor1} \\
    &= p(1)+q(1) - 2[p(1)q(1)]^{1/2} + O(p(1)q(1)) \label{eq:taylor2} \\
    &= (p(1)^{1/2} - q(1)^{1/2})^2 + O(p(1)q(1)), \label{eq:fine2}
\end{align}
where~\eqref{eq:taylor1} follows from the fact \( (1 - x)^{1/2} = 1 - \frac{x}{2} + O(x^2) \) for sufficiently small $x$, and~\eqref{eq:taylor2} follows from the Taylor series expansion $\log(1 - x) = -x + O(x^2)$.

Therefore, the threshold in~\eqref{eq:thre} becomes 
\begin{align}
\lim_{n \to \infty}\frac{n [p(1)^{1/2} - q(1)^{1/2}]^2}{\log n} = 2.
\end{align}

\subsubsection{Proof of Corollary~\ref{cor2}}

When the two distributions $p$ and $q$ satisfy~\eqref{eq:sp1} and~\eqref{eq:sp2}  respectively, we note that
\begin{align}
    I(p,q) = -2 \log\left[ (\alpha \beta)^{1/2} + [(1-\alpha)(1-\beta)]^{1/2}\right],
\end{align}
which is identical to~\eqref{eq:cor1}. Thus, the proof of Corollary~\ref{cor2} is the same as that of Corollary~\ref{cor1}.

\subsubsection{Proof of Corollary~\ref{cor3}}

When the two distributions respectively satisfy  $p(\bm{d}) = \prod_{i=1}^{D} p_i(d_i)$ and $q(\bm{d}) = \prod_{i=1}^{D} q_i(d_i)$, we have
\begin{align}
    I(p,q) &= -2 \log \left[ \sum_{\bm{d} \in \{0,1\}^D} [p(\bm{d}) q(\bm{d})]^{1/2} \right] \\
    &= -2 \log \left[ \sum_{\bm{d} \in \{0,1\}^D} \prod_{i=1}^D [p_i(d_i) q_i(d_i)]^{1/2} \right] \\
    &=  -2 \log \left[ \prod_{i=1}^D \sum_{d_i \in \{0,1\} } [p_i(d_i) q_i(d_i)]^{1/2} \right] \\
    &= \sum_{i=1}^D -2 \log \left[ \sum_{d_i \in \{0,1\} } [p_i(d_i) q_i(d_i)]^{1/2} \right] \\
    &= \sum_{i=1}^D -2 \log \left[ [p_i(0) q_i(0)]^{1/2} + [p_i(1) q_i(1)]^{1/2} \right].
\end{align}
Due to the definition of $\bar{p}$ and the assumption $\bar{p} = o(1)$, we have $p_i(1) = o(1)$ and $q_i(1) = o(1)$ for all $i \in [D]$. Thus, following the derivations from~\eqref{eq:fine}-\eqref{eq:fine2}, we have 
\begin{align}
    I(p,q) &= \sum_{i=1}^D  [p_i(1)^{1/2} - q_i(1)^{1/2}]^2 \notag \\
    &\qquad\quad+ O\left(\max_{i\in [D]} (p_i(1)+ q_i(1))^2 \right).
\end{align}

% Conclusion
Therefore, the threshold in~\eqref{eq:thre} becomes 
\begin{align}
\lim_{n \to \infty}\frac{n \sum_{i=1}^D [p_i(1)^{1/2} - q_i(1)^{1/2}]^2}{\log n} = 2.
\end{align}

\subsection{Proof of Lemma~\ref{lemma:change}}\label{sec:lemma2}
We aim to establish the relationship between \( \mathbb{E}_{\Phi}[\varepsilon_{\phi}(n)] \) and the distribution of \( R \) under \( \mathbb{P}_{\Psi} \). Specifically, for any function \( f(n) \), we have
\begin{align}
    &\mathbb{P}_{\Psi}\{R \leq f(n)\} \notag  \\
    &= \mathbb{P}_{\Psi}\{R \leq f(n), 1 \in \mathcal{N}\}  + \mathbb{P}_{\Psi}\{R \leq f(n), 1 \notin \mathcal{N} \}.
    \label{eq:PR}
\end{align}
Here, we decompose the probability $\mathbb{P}_{\Psi}\{R \leq f(n)\}$ into two terms. The first term corresponds to the event where node $1$ belongs to the set $\mathcal{N}$, and the second term corresponds to the complementary event where node $1$ does not belong to $\mathcal{N}$. 
Next, we examine the two terms \(\mathbb{P}_{\Psi}\{R \leq f(n), 1 \in \mathcal{N}\}\) and \(\mathbb{P}_{\Psi}\{R \leq f(n), 1 \notin \mathcal{N}\}\) separately.

\subsubsection{Case 1 (when node $1$ belongs to $\mathcal{N}$)}
Based on the definition of $R$, we have
\begin{align}
&\mathbb{P}_{\Psi} (R \le f(n), 1 \in \mathcal{N}) \notag \\
&=\mathbb{P}_{\Psi}\left( \sum_{v=2}^n \log \frac{\mathbb{P}_{\Psi}(\bm{A}_{1v}) }{\mathbb{P}_{\Phi}(\bm{A}_{1v})} \le f(n), 1 \in \mathcal{N} \right) \nonumber \\
&= \sum_{ \{\bm{d}_{1v} \in \{0,1\}^D\}_{v=2}^n } \left(  \prod_{v =2}^n \mathbb{P}_{\Psi} (\bm{A}_{1v} = \bm{d}_{1v}) \right) \notag \\
&\qquad\quad\times \bm{1}\left\{ \prod_{v=2}^{n} \frac{ \mathbb{P}_{\Psi} (\bm{A}_{1v} = \bm{d}_{1v})}{ \mathbb{P}_{\Phi} (\bm{A}_{1v} = \bm{d}_{1v})} \leq \mathrm{e}^{f(n)}, 1 \in \mathcal{N} \right\}  \nonumber \\
&\le \sum_{ \{\bm{d}_{1v} \in \{0,1\}^D\}_{v=2}^n } e^{f(n)} \left(  \prod_{v =2}^n \mathbb{P}_{\Phi} (\bm{A}_{1v} = \bm{d}_{1v}) \right) \notag \\
&\qquad\quad\times \bm{1}\left\{ \prod_{v=2}^{n} \frac{ \mathbb{P}_{\Psi} (\bm{A}_{1v} = \bm{d}_{1v})}{ \mathbb{P}_{\Phi} (\bm{A}_{1v} = \bm{d}_{1v})} \leq \mathrm{e}^{f(n)}, 1 \in \mathcal{N}  \right\}  \label{eq:measure} \\
&= e^{f(n)} \cdot \mathbb{P}_{\Phi} (R \le f(n), 1 \in \mathcal{N}),
\label{eq:P_phi1}
\end{align}
where in~\eqref{eq:measure} we change the measure from $\mathbb{P}_{\Psi}(\cdot)$ to $\mathbb{P}_{\Phi}(\cdot)$.
By noting that 
\begin{align}
\mathbb{P}_{\Phi} (R \le f(n), 1 \in \mathcal{N}) &\le \mathbb{P}_{\Phi} (1 \in \mathcal{N}) = \frac{\mathbb{E}_{\Phi}(\varepsilon_{\phi}(n))}{n},
\end{align}
we obtain 
\begin{align}
\mathbb{P}_{\Psi} (R \le f(n), 1 \in \mathcal{N}) \le e^{f(n)} \frac{\mathbb{E}_{\Phi}(\varepsilon_{\phi}(n))}{n}.
\end{align}

\subsubsection{Case 2 (when node $1$ does not belong to $\mathcal{N}$)} Recall that the two communities are denoted by $C_+$ and $C_-$. Let $\widehat{C}_+$ and $\widehat{C}_-$ be the estimated communities. 
We evaluate \(\mathbb{P}_{\Psi}\left\{R \leq f(n), 1 \notin \mathcal{N} \right\}\) as follows:
\begin{align}
& \mathbb{P}_{\Psi}\left\{R \leq f(n), 1 \notin \mathcal{N} \right\} \notag \\ 
& \leq \mathbb{P}_{\Psi}\left\{1 \notin \mathcal{N} \right\} \notag \\ 
& =\frac{1}{2} \mathbb{P}_{\Psi}\{1  \in \widehat{C}_{+} \mid  1 \in C_+ \}+\frac{1}{2} \mathbb{P}_{\Psi} \{ 1 \in \widehat{C}_{-} \mid 1 \in C_- \} \notag \\ 
&=\frac{1}{2} \mathbb{P}_{\Psi}\{1  \in \widehat{C}_{+} \mid  1 \in C_+ \}+\frac{1}{2}  \left( 1 \!-\! \mathbb{P}_{\Psi}\{ 1 \in \widehat{C}_{+} \! \mid 1 \in C_- \} \right) \notag \\
& = \frac{1}{2}. \notag
\end{align}
The last equation follows from the fact that $\mathbb{P}_{\Psi}\{1 \in \widehat{C}_{+} \mid 1 \in C_+ \}=\mathbb{P}_{\Psi}\{1 \in \widehat{C}_{+} \mid 1 \in C_- \}$, which is due to the definition of the model $\Psi$.

%We begin by noting that in the model \(\Psi\) with multi-views, the joint probability distribution is invariant to whether the nodes \(v^{\star}\) belong to random communities. Mathematically, this is represented as:
%This Eqn.~\eqref{eq:P_phi4} shows that \(\mathbb{P}_{\Psi}\left\{R \leq f(n), v^{\star}\notin\varepsilon\right\}\) is bounded above by \( \frac{1}{2} \).
Finally, by combining Case 1 and Case 2, we conclude that
\begin{equation}
\mathbb{P}_{\Psi}\{R \leq f(n)\} \leq \mathrm{e}^{f(n)} \cdot \frac{\mathbb{E}_{\Phi}\left[\varepsilon_{\phi}(n)\right]}{n}+\frac{1}{2}.
\label{eq:P_phi5}
\end{equation}
This completes the proof of Lemma~\ref{lemma:change}.

\section{Conclusion}
 We introduce the multi-view stochastic block model (MVSBM) for generating multiple potentially correlated graphs on the same set of nodes, while most previous studies focus on the problem setting of a single graph. Technically, we study community detection in the MVSBM from an information-theoretic perspective. We first propose using the ML estimator to detect the hidden communities and then provide a performance bound for the ML estimator. Subsequently, we establish a lower bound (on the expected number of misclassified nodes) that serves as an impossibility result for any estimator. Notably, we show that this lower bound coincides with the aforementioned performance bound, allowing us to reveal a sharp threshold for community detection in the MVSBM. Moreover, we highlight that our results can be applied to various specific cases of the MVSBM, offering valuable insights for many practical applications.

In the future, it will be interesting to expand our theory to more general settings and some variations of the MVSBM, such as MVSBM with overlapping communities, weighted or labeled MVSBMs, MVSBM with side information, etc. Additionally, exploring the application of MVSBM to different types of data, such as text, images, or biological networks, could provide valuable insights and further validate the versatility and effectiveness of the model.

\vspace{10pt}

\bibliographystyle{IEEEtran}

\bibliography{MVSBM}

% \end{IEEEbiography}



\vfill

\end{document}